\documentclass[12pt,authoryear]
{article}
\usepackage[margin=2cm]{geometry}

\usepackage{authblk}

\usepackage{algorithm}
\usepackage{algpseudocode}
\usepackage{tikz}

\usepackage{natbib}
 \bibpunct[, ]{(}{)}{,}{a}{}{,}%
\usepackage{amsmath}
\usepackage{amssymb}
\usepackage{amsfonts}
\usepackage{amsthm}
\usepackage{mathtools}
\usepackage{bbm}
\usepackage{tikz}
\usetikzlibrary{decorations.pathreplacing,decorations.markings}

\mathtoolsset{%
}
\usepackage{relsize}

\usepackage{acronym}
\usepackage{enumitem}
\usepackage{longtable}
\usepackage{tikz}
\usetikzlibrary{decorations.pathreplacing,decorations.markings}

\usepackage{hyperref}
\hypersetup{
colorlinks=true,
linktocpage=true,
pdfstartview=FitH,
breaklinks=true,
pdfpagemode=UseNone,
pageanchor=true,
pdfpagemode=UseOutlines,
plainpages=false,
bookmarksnumbered,
bookmarksopen=false,
bookmarksopenlevel=1,
hypertexnames=true,
pdfhighlight=/O,
nesting=true,
frenchlinks, urlcolor=red!80!black,linkcolor=green!60!black,citecolor=blue!80!black, 
pdftitle={},
pdfauthor={},
pdfsubject={},
pdfkeywords={},
pdfcreator={pdfLaTeX},
pdfproducer={LaTeX with hyperref}
}

\numberwithin{equation}{section}  
\usepackage[sort&compress,capitalize,nameinlink]{cleveref}

\usepackage[textwidth=30mm]{todonotes}

\DeclarePairedDelimiter{\braces}{\{}{\}}
\DeclarePairedDelimiter{\bracks}{[}{]}
\DeclarePairedDelimiter{\parens}{(}{)}

\DeclarePairedDelimiter{\abs}{\lvert}{\rvert}

\theoremstyle{plain}
\newtheorem{theorem}{Theorem}
\newtheorem{corollary}[theorem]{Corollary}
\newtheorem{lemma}[theorem]{Lemma}

\theoremstyle{definition}
\newtheorem{definition}[theorem]{Definition}

\theoremstyle{remark}
\newtheorem{remark}[theorem]{Remark}
\newtheorem{example}[theorem]{Example}

\numberwithin{theorem}{section}

\newcommand{\debug}[1]{#1}

\newcommand{\ie}{i.e., }
\newcommand{\eg}{e.g., }
\newcommand{\iid}{i.i.d.\ }

\newcommand{\naturals}{\mathbb{\debug N}}

\DeclareMathOperator{\Expect}{\debug{\mathsf{E}}}
\DeclareMathOperator{\Prob}{\debug{\mathsf{P}}}

\newcommand{\alice}{\mathsf{\debug A}}

\newcommand{\arrivr}{\debug \lambda}

\newcommand{\bob}{\mathsf{\debug B}}

\newcommand{\cost}{\debug c}
\newcommand{\crit}{\mathsf{\debug c}}
\newcommand{\ctime}{\debug t}

\newcommand{\differ}{\debug D}
\newcommand{\differtilde}{\widetilde{\differ}}
\newcommand{\dirac}{\debug \delta}

\newcommand{\labcust}{\debug i}
\newcommand{\labcustalt}{\debug j}
\newcommand{\labcuststar}{\debug{\labcust^{\ast}}}
\newcommand{\labels}{\debug I}
\newcommand{\labserv}{\debug \kappa}

\newcommand{\maxnjobs}{\debug N}
\newcommand{\MMS}[1][\nserv]{M/M/$#1$\,}

\newcommand{\ngparameters}{\parens*{\bar{\arrivr},\parens*{\bar{\servr}_{\labserv}}_{\labserv\in[\nserv]},\bar{\cost},\bar{\reward}}}
\newcommand{\njobs}{\debug n}
\newcommand{\njobsalt}{\debug m}
\newcommand{\njobsA}{\njobs_{\alice}}
\newcommand{\njobsB}{\njobs_{\bob}}
\newcommand{\nopt}{\njobs^{\ast}}

\newcommand{\nserv}{\debug s}

\newcommand{\parameters}{\parens*{\arrivr,\parens*{\servr_{\labserv}}_{\labserv\in[\nserv]},\cost,\reward}}
\newcommand{\parametersone}{\parens*{\arrivr,\servr,\cost,\reward}}
\newcommand{\posa}{\debug \pi}
\newcommand{\posit}{\debug p}

\newcommand{\qlthresh}{\debug Y}
\newcommand{\regime}{\debug R}
\newcommand{\return}{\mathsf{\debug r}}
\newcommand{\reward}{\debug r}
\newcommand{\round}{\debug k}
\newcommand{\ruinp}{\debug \theta}
\newcommand{\ruinptilde}{\widetilde{\ruinp}}
\newcommand{\ruint}{\debug T}
\newcommand{\ruinttilde}{\widetilde{\ruint}}

\newcommand{\score}{\debug \beta}
\newcommand{\servr}{\debug \mu}
\newcommand{\setA}{\debug A}
\newcommand{\simplex}{\debug \Delta}
\newcommand{\slot}{\debug z}
\newcommand{\states}{\mathcal{\debug X}}
\newcommand{\statex}{\debug x}
\newcommand{\stime}{\debug \tau}
\newcommand{\stimetilde}{\widetilde{\stime}}
\newcommand{\strat}{\debug \sigma}
\newcommand{\stratopt}{\debug{\strat^{\ast}}}

\newcommand{\transa}{\debug \alpha}
\newcommand{\transr}{\debug \rho}
\newcommand{\transs}{\debug \xi}

\newcommand{\acdef}[1]{\acfi{#1}\acused{#1}}

\newacro{FCFS}{first-come-first-served}
\newacro{FCLS}{first-come-last-served} 
\newacro{LCFS}{last-come-first-served}
\newacro{LCFSPR}[LCFS-PR]{last-come-first-served with preemption}
\newacro{PS}{priority-slots}
\newacro{SIRO}{service in random order}

\begin{document}

\title{A Characterization of Universally Optimal Queueing Regimes}

\author[1]{Marco Scarsini}
\affil[1]{Department of Economics and Financial Markets, Luiss University\\
\texttt{marco.scarsini@luiss.it}
}

\author[2]{Eran Shmaya}
\affil[2]{Department of Economics, Stony Brook University\\
\texttt{eran.shmaya@stonybrook.edu}
}

\maketitle

\begin{abstract}
We consider an \MMS queueing model in which customers strategically decide, based on the service reward and waiting cost, whether to join upon arrival or balk and, at any time, whether to remain in the queue or renege. 
Rational strategic behavior yields an equilibrium whose outcome may be socially efficient or inefficient, depending on the queueing regime. 
Some regimes yield an efficient equilibrium only under precise calibration to the model parameters. 
Others are universally optimal, meaning that their equilibrium outcome is efficient for all parameter values. 
Universal optimality is therefore an appealing property for a planner choosing a queueing regime. 
We characterize the class of universally optimal queueing regimes. 
A by-product of our characterization is that preemption plays an unavoidable role in universally optimal regimes.


\end{abstract}



%
%

\section{Introduction}
\label{se:intro}

Queues are often designed rather than given: a platform, clinic, or service provider chooses an admission rule (who is allowed to join), an ordering rule (who is served next), and possibly an expulsion rule (who is asked to leave). 
When customers can strategically decide whether to join and whether to abandon, these design choices shape both congestion and incentives.
A foundational example is the  observable \MMS[1] queueing model of \citet{Nao:E1969}, where customers choose whether to enter the queue or balk. 
Under the standard \ac{FCFS} regime, equilibrium entry is typically excessive relative to the social optimum because each entrant imposes a waiting-time externality on those who arrive later. 

Classic remedies—entry tolls or a cap on the queue length—can implement the social optimum, but they require fine calibration to the primitives (arrival and service rates, reward from service, and waiting costs). 
A striking alternative is to change the queueing regime instead of the price. 
\citet{Has:E1985} showed that, in the same \MMS[1] environment, \ac{LCFSPR} yields an equilibrium that is socially optimal without parameter tuning. 
The economic intuition is simple: under \ac{LCFSPR}, the customer who is last internalizes the marginal congestion cost.
In fact, this last customer's decision to stay does not delay anyone else: future arrivals jump ahead, and current customers are already ahead. 
More broadly, Hassin’s insight extends to a sufficient structural condition: if the regime never assigns a newly arriving customer the last position, then equilibrium behavior implements the social optimum, uniformly across parameters. 

More recent work uncovered additional robustly efficient regimes that do not fit the simple never-last rule. 
In particular, the priority-slots regime \citep{Wan:MS2016,HavOz:ORL2016} assigns each joining customer to the best available slot, and serves the best occupied slot at each time; this regime is also universally optimal. 
In priority-slots, however, when slots $1,\dots,\njobs$ are occupied, the next arriving customer takes slot $\njobs+1$, which is (by design) the last position. 

This raises a natural design question:
which queueing regimes are universally optimal? 
That is, which regimes implement the socially optimal outcome in equilibrium for every set of primitives, without requiring the designer to know or estimate those primitives?
This paper answers that question for an observable \MMS system with heterogeneous servers. 
There are two main contributions.

First, we introduce an abstract definition of a queueing regime that encompasses standard disciplines (\ac{FCFS}, \ac{LCFS} variants, random order) as well as richer mechanisms such as priority-slots. 

Second, we provide a complete characterization of universal optimality. 
The starting point is Hassin’s sufficient condition---never place an arrival last---which is not necessary in general: for instance, priority-slots is universally optimal despite sometimes placing arrivals last. 

Our main result shows that the correct refinement is the following:
A newly arriving customer may be placed in the last position only at special states that cannot be reached from a larger system. 
We formalize these special states as maximal states: informally, a state is maximal if it cannot be reached from a larger system by a sequence of arrivals and services that never serves the last customer. 

This definition captures, for example, why priority-slots can place arrivals last exactly when the occupied slots are $[\njobs]$: those are precisely the states at which the mechanism ``resets'' who bears the marginal congestion cost.
This characterization also yields an intuitive implementation perspective. The social optimum in these Markovian queueing environments takes a threshold form: admit and retain customers up to an optimal cutoff $\njobs^{*}$, and expel (or induce abandonment of) those beyond it. 
The designer, however, does not know $\njobs^{*}$ because it depends on unknown primitives. 
Our characterization implies that a universally optimal regime effectively delegates the marginal decision---whether the cutoff should be $\njobs$ or $\njobs+1$---to a customer who is placed last at the right decision points (the maximal states). 
Outside those decision points, universal optimality forces the mechanism to prevent arrivals from becoming last, preserving the key externality-free property behind Hassin’s insight. 

Finally, our characterization delivers two additional implications. 
First, it formalizes a common—but previously informal—intuition: any universally optimal regime must sometimes interrupt a customer's service (our notion of ``preemption'') in at least one non-idle state. 
Second, it provides a useful dominance property: under any regime, if it is socially optimal for a customer to stay, then staying is also rational. 
We close by illustrating our condition with new examples of universally optimal regimes. 

The notion of universal optimality is closely related to implementation under complete information: agents know the primitives, whereas  the designer commits to a regime that implements the social optimum across all primitives.



%
%

\subsection{Related Work}
\label{suse:literature-review}
For a survey of the literature on strategic queueing we refer the reader to  \citet{HasHav:Kluwer2003} and  \citet{Has:CRC2016}. 
Among other things, the literature deals with heterogeneous customers, general arrival and service processes, and waiting-cost functions. 
Not long after \citet{Nao:E1969}, several  papers considered variations on the topic.
For instance, 
\citet{Yec:OR1971,Yec:MS1972} generalized Naor's model to GI/M/1 and GI/M/$s$ queueing systems, respectively, and  
\citet{Knu:E1972} studied an \MMS model with a general cost/benefit function.
\citet{Sim:SSMH1976} studied a GI/M/$s$ queueing system and compared the equilibrium and optimal thresholds in the classical case and in the case where a toll collector maximizes the total revenue.
To achieve his results, he used the technique of the ``polite customer,'' \ie a customer who will be served only when no other customer is waiting, and will be preempted, if necessary.
A further generalization was provided by \citet{Sti:IEEETAC1985}, who considered a  model where the cost per unit time is a convex, nondecreasing function of the number of customers present in the queue. 

Whereas Naor's model assumed that customers observe the queue length, a large strand of the literature following  \citet{Lit:IJPR1974} and \citet{EdeHil:E1975} studied  settings in which customers do not observe the current state of the queue. 

\citet{Has:E1985} considered queueing regimes that, for any possible values of the relevant parameters, guarantee optimality in equilibrium. 
Hassin's claim was extended to an \MMS model by \citet{XuSha:OR1993}, using what they call the \emph{dual approach}:
They related optimal admission control in an \MMS queue under  \ac{FCFS} to the dual expulsion-control problem under \ac{LCFSPR}.
This is similar to \citet{Sim:SSMH1976}, where a present customer does not produce any externality on future customers.
\citet{Wan:MS2016} used the dual approach to propose a new queueing regime under which the equilibrium behavior of customers is socially optimal.  
This regime was then analyzed and simplified by \citet{HavOz:ORL2016}, who also proposed a new regime that guarantees social optimality of equilibria and does not require preemption. 
Unfortunately, this regime is not robust to changes of the model parameters. 
\citet{HavOz:QMSM2018} introduced an interesting criterion, called \emph{social cost of deviation}, and used it to give new proofs of some classical results, including the ones by  \citet{Nao:E1969,EdeHil:E1975}, and \citet{Has:E1985}.

Various papers compared the equilibrium outcome and the social optimum in strategic queues, under different assumptions on the stochastic model and on the cost and reward structure \citep[see, \eg][]{LipSti:OR1977,Sti:MS1978,BelSti:MS1983,GuoHas:OR2011,HasSni:OR2020,WanPraHanHas:SS2024,ZhoGopWar:OR2025}.
The optimal policy for an \MMS system with heterogeneous servers was studied by, among others,  \citet{KimAhnRig:JAP2011,Aga:IP2021}. 
\citet{AfeSar:SSRN2026,AfeHuIbrSar:SSRN2026} and \citet{DAnSca:arXiv2025} studied strategic queueing models where customers are split into two classes with different priorities.

\citet{HavOz:ORL2016} studied various protocols that guarantee optimality and showed that, to implement either \ac{LCFS} or priority-slots, the social designer does not need to know the parameters of the model.    
To our knowledge, existing papers study equilibrium behavior and social optimum under a given regime, but characterizations of universally optimal regimes do not exist.

\citet{SuZen:MSOM2004,SuZen:OR2005,ShiYin:SSRN2022} compared the efficiency of \ac{FCFS} and \ac{LCFS} regimes in the framework of organ transplants.
The inefficiency of the equilibrium in the Naor model was quantified by \citet{GilHasKer:IEEETAC2014} using the price of anarchy. 
\citet{GilHas:ORL2014} considered an \MMS[1] model where a fraction of the population cooperates with a central planner.
In the context of ridesharing platforms, \citet{CasMaNazYan:arXiv2021} considered a randomized variation of the \ac{FCFS} regime and showed how it improves efficiency.
\citet{FelSeg:MS2022} considered limits to the time that customers spend in the queue as a way to control congestion.
In a behavioral study, \citet{Bue:MS2021} examined customers' aversion to being placed at the end of a queue and analyzed its implications for abandonment.
\citet{AfeAkaBenDabDen:SSRN2025} studied a two-sided observable queue where strategic agents arrive on each side with the goal of matching.

Some recent papers in economics have dealt with issues related to strategic queues. 
In \citet{CriTho:TE2019} customers learn about the server's quality. 
In \citet{Mar:TE2025} customers learn about their own valuation from service. 
Perhaps the closest paper to ours is  \citet{CheTer:JPE2025}, who considered a design problem with a general class of queueing regimes. 
There are several major differences between their paper and ours:
First, in their paper the designer tailors the regime to the parameters, whereas we consider universally optimal regimes. 
Second, in their paper the designer's payoff is a combination of the customers' welfare and the provider's payoff, whereas in our paper the designer only cares about the customers' payoff. 
Third, in their paper the customers do not observe the state of the queue. 
Their optimal regime is \ac{FCFS} with a cap (which depends on the parameters); this is not universally optimal by our definition. 
The class of regimes they study is also different from ours, and we explain the difference later. 
\citet{DovSze:GEB2025} extended Naor's argument to a two-sided matching market. 
Their paper includes references to other recent papers that deal with welfare implications of dynamic matching.

%
%

\subsection{Notation}
\label{suse:notation}

The symbol $[\njobs]$ denotes the set $\braces*{1,\dots,\njobs}$.
Given a finite set $\setA$, the symbol $\simplex(\setA)$ denotes the set of probability measures of $\setA$.

%
%

\subsection{Organization of the paper}
\label{suse:organization}

The paper is organized as follows: 
\cref{se:environment} describes the main concepts of the model.
\cref{se:optimal-regime} states the main results about universally optimal regimes.
\cref{se:proof} is devoted to the proofs.
\cref{se:conclusions} offers some concluding remarks.
\cref{se:symbols} contains a list of symbols used throughout the paper. 

%
%

\section{Environment}
\label{se:environment}

We consider an \MMS queueing system, \ie a system  with $\nserv$ servers, where customers arrive according to a Poisson process with rate $\arrivr$, and, for $\labserv\in[\nserv]$, the service times for server $\labserv$ are \iid exponentially distributed random variables with rate $\servr_{\labserv}$,  and service times are independent across servers. 
For the sake of simplicity, without  loss of generality, we assume that $\servr_{\labserv} \ge \servr_{\labserv+1}$, for $\labserv\in\braces*{1,\dots,\nserv-1}$.
Each customer incurs a flow cost rate $\cost$ while in the system, and receives a reward $\reward$ upon service completion.

%
%

\subsection{Social Optimum}
\label{se:social-optimum}

The social designer cares about the (undiscounted) average welfare of the customers in the long run. 
Equivalently, the designer incurs a flow cost rate $\cost$ per customer in the system and receives a reward $\reward$ upon each customer's service completion. 
The designer does not care about the identity of the customers who are being served, and can decide which arriving customers to accept, which server will serve which customer, and when to kick existing customers out of the system.  

Because the environment is Markovian, the socially optimal policy prescribes an intervention only when a new customer arrives or when a customer completes a service. 
In an \MMS[1] context,  \citet{Nao:E1969} proved that, to achieve the social optimum, the designer should admit customers into the system as long as their number is at most a threshold $\nopt$, which depends on the parameters $\parametersone$, and remove customers from the system if the number of customers exceeds that threshold. 
From the social designer's point of view, it does not matter which customer is removed: it can be the newly arriving customer or a customer who was already  waiting in the queue. 
This analysis was extended to the \MMS case by \citet{XuSha:OR1993}.

%
%

\subsection{Queueing Regimes}
\label{se:anonymous-regime}

We now provide our definition of queueing regime. 
A \emph{queueing regime} $\regime$ is given by a tuple $\parens*{\states, \transa,\parens*{\transs_{\labserv}}_{\labserv\in[\nserv]},\parens*{\transr_{\labcust}}_{\labcust \ge 1},\posa}$, where $\states$ is a countable set of \emph{states},  $\transa,\transs_{\labserv},\transr_{\labcust}$ are  \emph{transition functions}, and $\posa$ is a \emph{position function}.    
The set of states can be partitioned as $\states=\states_{0}\uplus \states_{1}\uplus\dots$, where $\states_{\njobs}$ is the set of possible states with $\njobs$ customers  in the system, and $\uplus$ is the disjoint union. 
We also assume that $\states_{0}$ is a singleton, representing the  \emph{idle} system. 
For $\statex\in\states_{\njobs}$ we  define  $\njobs(\statex)=\njobs$. 
Throughout the paper, we assume that, at any given time, the state of the system can be observed by all customers.

At any point in time, the customers who are currently in the system are ranked according to some order. 
This order, which we refer to as \emph{rank}, is the order in which they will be served if no new customer joins and nobody reneges.

To satisfy optimality, we assume that, at any given time, the customer with rank $\labcust$ is served by server $\labcust$, for $\labcust \le \min(\njobs,\nserv)$. 
Thus, better servers are always the first to be active. 
If $\njobs>\nserv$, then customers ranked $\nserv+1,\nserv+2,\dots,\njobs$ are in the queue in this order.
This model can equivalently be represented as an \MMS[\infty] model with $\servr_{\labserv}=0$, for all $\labserv>\nserv$, and \emph{reassignment}.  
Reassignment means that, if a newly arriving customer is assigned rank $m\in[\njobs+1]$, then all customers with rank smaller than $m$ retain their rank, whereas the rank of all other customers increases by $1$.
Similarly, if the customer with rank $m$ is either served or reneges, then all customers with rank smaller than $m$ retain their rank, whereas  the rank of all other customers decreases by $1$ \citep[see, \eg][]{AkgDowRig:IEEETAC2014}.
When a customer's service is interrupted and the customer is sent back to the queue, we refer to this as  \emph{preemption}. 
This occurs when the customer in position $\nserv$ is reassigned to position $\nserv+1$. 
Under the above assumptions, if one server (any server) is available, no customer has an incentive to remain in the queue, rather than being served by this server. 
In a regime that never uses  reassignment, customers may prefer to stay in the queue and wait for a fast server, rather than being served by a slow server.

The system transitions from one state to another when either a new customer arrives, or a customer is served, or a group of customers (possibly only one) renege.
Arrivals and service completions  are random and controlled by Nature, whereas reneging is a decision made by the customer. 
We assume that none of these events changes the relative order among the existing customers in the system.

The transition rules of the system and the rank of new customers in the system are governed by the transition functions $\transr_{\labcust},\transs_{\labserv},\transa$, and the position function $\posa$ as follows: 
\begin{itemize}
\item 
If the system is in state $\statex\in \states_{\njobs}$ and a new customer arrives,  the system transitions to state  $\transa(\statex)\in\states_{\njobs+1}$ and the arriving customer is assigned a random rank, drawn from the distribution $\posa(\statex)\in \simplex([\njobs+1])$. 

\item If the system is in state $\statex\in \states_{\njobs}$ with $\njobs\ge 1$ and the customer who is being served by server $\labserv\in[\nserv]$ completes service, the system transitions to state $\transs_{\labserv}(\statex)\in\states_{\njobs-1}$.
\item
If the system is in state  $\statex\in \states_{\njobs}$ and the customer whose current position is $\labcust\in [\njobs]$ reneges, the system transitions to state $\transr_{\labcust}(\statex)\in\states_{\njobs-1}$.
\end{itemize}
For $\labcust<\labcustalt$, we assume that, if the customers who are currently in positions $\labcust$ and $\labcustalt$ renege, then the new state does not depend on the order in which they renege, \ie $\transr_{\labcust}(\transr_{\labcustalt}(\statex))=
\transr_{\labcustalt-1}(\transr_{\labcust}(\statex))$.
This allows us to talk unambiguously about simultaneous reneging: for $\labels=\parens*{\labcust_{1} < \labcust_{2} < \dots < \labcust_{\round}}$, we let $\transr_{\labels} \coloneqq \transr_{\labcust_{1}} \circ \transr_{\labcust_{2}} \circ \dots \circ \transr_{\labcust_{\round}}$.

\begin{remark}
\label{re:balk-renege}
Up to now, we made no distinction between balking and reneging. 
The term balking refers to an arriving customer who chooses not to enter the system at all. 
In our notation, this corresponds to a customer who reneges immediately upon arrival. 
In some situations, it is natural to add the assumption that the state of the system does not change if a customer balks (for example, because the system does not observe balking customers). 
We do not make this assumption in the paper, because our theorem holds without it.
\end{remark}

In a recent paper, \citet{CheTer:JPE2025} also considered a general definition of queueing regime. 
On the one hand, their definition is more permissive than ours, in that it allows for splitting the service rate between the customers who are currently in the queue. 
On the other hand, it is more restrictive than ours because it does not allow the positions of arriving customers in the queue to depend on a Markov state as our definition does. 
For example, the priority-slots  regime does not belong to their class of regimes. 

We now provide a few examples of regimes. 
In the sequel, when the new position is deterministic, we sometimes abuse notation and,  for $\labcust\in [\njobs+1]$, we write $\posa(\statex)=\labcust$ instead of  $\posa(\statex)=\dirac_{\labcust}$ (the Dirac distribution on $\labcust$).

\begin{example}[\Acl{FCFS}]
\label{ex:FCFS-s}

The  \acdef{FCFS} regime is defined as follows: The state only encodes the number of customers in the system; hence, (up to renaming) 
$\states_{\njobs}=\{\njobs\}$.
The transition and position functions are:
\begin{equation}\label{eq:transition-FCFS}
\begin{aligned}
\transa(\njobs)&=\njobs+1,\\
\transs_{\labserv}(\njobs)&=\njobs-1,\\
\transr_{\labcust}(\njobs)&=\njobs-1,\\
\posa(\njobs)&=\njobs+1.
\end{aligned}
\end{equation}
\end{example}

\begin{example}
[\Acl{LCFS}]
\label{ex:LCFS-s}

The  \acdef{LCFS} regime has the same state space and the same transition functions as  \ac{FCFS} in \eqref{eq:transition-FCFS}. 
In the \acdef{LCFSPR} regime  $\posa(\statex)=1$ for every state $\statex$.
In the  \ac{LCFS} \emph{without preemption} $\posa(\statex)=\min(\nserv+1,\njobs(\statex)+1)$ for every state $\statex$. 
\end{example}

\begin{example}
[\Acl{SIRO}]
\label{ex:random-order-s} 

The  \acdef{SIRO} regime has the same state space and the same transition functions as in the previous two examples. 
The position function $\posa(\njobs)$ assigns probability $1/(\njobs+1)$ to each element of $[\njobs+1]$, so that a newly arriving customer is assigned a random rank.
\end{example}

\begin{example}[\Acl{PS}]
\label{ex:priority-slots-s}

In the \acdef{PS} regime, there is a countable set  $\naturals$ of slots and the state space is given by the set of occupied slots, so an element of $\states_{\njobs}$ is a subset of $\naturals$ of cardinality $\njobs$. 
If $\statex=\{\slot_{1},\dots,\slot_{\njobs}\}\in\states_{\njobs}$ with $\slot_{1} < \dots < \slot_{\njobs}$, then 
\begin{equation}
\begin{aligned}
\posa(\statex)
&=\min (\naturals\setminus \statex),\\
\transa(\statex)
&=\statex\cup \{\posa(\statex)\},\\
\transs_{\labserv}(\statex)
&= \statex\setminus\{\slot_{\labserv}\}\\
\transr_{\labcust}(\statex)
&=\statex\setminus\{\slot_{\labcust}\}.
\end{aligned}
\end{equation}
\end{example}

%
%

\subsection{Strategies and Equilibrium}
\label{suse:strategies}

A \emph{Markov strategy profile} is a function $\strat$ defined over non-idle states such that $\strat(\statex)\subseteq[\njobs]$ for every $\statex\in \states_{\njobs}$, 
with the interpretation that $\strat(\statex)$ are the positions of players who renege in state $\statex$.
We assume that reneging occurs simultaneously whenever the system reaches this state. 

The social optimum is achieved by a strategy profile $\strat$ such that 
$\abs*{\strat(\statex)}=\max\parens*{\njobs-\nopt\parameters, 0}$ for every $\statex\in \states_{\njobs}$. 
This was proved by \citet{Nao:E1969}  for  \MMS[1] and by \citet{XuSha:OR1993} for \MMS queueing systems \citep[see also][]{AkgDowRig:IEEETAC2014,Aga:IP2021}.
The strategy profile $\strat$ is said to induce the socially optimal behavior.

A Markov strategy profile is a \emph{Markov perfect equilibrium} if, for every state $\statex$, the profile is a Nash equilibrium in the game that starts in state $\statex$, in which players can decide whether to stay in the system or to renege. 

Note that the game may have equilibria with individual strategies that are not deterministic and do not renege immediately after a transition to a new state (for example, there can be a war of attrition among players' decision to renege). 
However, these equilibria, if they exist, cannot induce a socially optimal behavior.

%
%

\section{Universally Optimal Regimes}
\label{se:optimal-regime}

A regime is \emph{universally optimal} if, for every parameter vector  $\parameters$, the game admits a Markov perfect equilibrium that induces the socially optimal behavior. 
Our goal is to characterize the class of universally optimal regimes.
To do so, we need the following definition.

\begin{definition}
\label{de:maximal-state}
A state $\statex\in\states_{\njobs}$ is \emph{maximal} if there does not exist a sequence $\statex_{0},\statex_{1},\dots,\statex_{\round}=\statex$ of states such that $\njobs(\statex_{0}) > \njobs$, and for every $1\le \labcustalt \le \round$,
either $\statex_{\labcustalt}=\transa(\statex_{\labcustalt-1})$ or 
$\statex_{\labcustalt}=\transs_{\labserv}(\statex_{\labcustalt-1})$ for some $\labserv < \njobs(\statex_{\labcustalt-1})$. 
\end{definition}
In words, a state is maximal if it cannot be reached from a state with a larger number of customers by a finite sequence of arrivals and services of customers who are not the last in the system (in particular, an idle state is always maximal).
 
\begin{example}[\Acl{FCFS}, \Acl{LCFS}, \Acl{SIRO}]
\label{ex:FCFS-2}
In these regimes the state is the number of customers in the system. 
The only maximal state is the idle state.
Indeed, for $\njobs>0$, it is possible to have   $\njobs$ customers in the system now and to have had $\njobs+1$ customers in the past.
For $\njobs=0$, this is possible only if the last customer is served.
\end{example}

\begin{example}[Priority-slots]
\label{ex:priority-slots-2}
Recall that a state is given by the set of occupied slots. 
We claim that a state $\statex\in\states_{\njobs}$ is maximal if and only if $\statex=[\njobs]$, that is, the slots that are occupied are exactly $1,\dots,\njobs$. 

Indeed, to show that $\statex=[\njobs]$ is maximal, assume we reached it from state $\statex_{0}$ with $\njobsalt\ge\njobs$ customers without reneging and without serving the last customer in the system.
Let $\slot\ge \njobsalt$ be the worst slot that was ever occupied during the transition from $\statex_{0}$ to $\statex$. 
By the way the  mechanism is defined,  slot $\slot$ must still be occupied in state $\statex$ because the customer at that slot is always the last in the system.
Therefore, $\slot\le \njobs$, which implies $\njobsalt=\njobs$. 
This proves that $\statex$ is maximal. 

Assume now that $\statex \in \states_\njobs$ and $\statex\neq[\njobs]$ for any $\njobs$. 
Then there exists some $\njobsalt>\njobs$ such that $\statex = [\njobsalt]\setminus\{\slot_{1},\dots,\slot_{\round}\}$ with $\slot_{1}>\dots>\slot_{\round}$ and $\round=\njobsalt-\njobs$. 
Consider the sequence of states that starts with $\statex_{0}=[\njobsalt]$,  followed by $\slot_{1}$ services of the first customer in the system, then $\slot_{1}-1$ arrivals, then $\slot_{2}$ services  of the first customer in the system followed by $\slot_{2}-1$ arrivals, ending with $\slot_{\round}$ services  of the first customer in the system followed by $\slot_{\round}-1$ arrivals. 
By the way the mechanism is defined, this sequence ends in state $\statex$ without serving the customer at slot $\njobsalt$, who is always the last in the system.
Because $\njobsalt>\njobs$, this proves that $\statex$ is not maximal.
\end{example}

The following theorem states our characterization.

\begin{theorem}
\label{th:thetheorem}
The following two conditions are equivalent for a queueing regime:
\begin{enumerate}[label=\emph{(\alph*)}, ref=(\alph*)]
\item 
\label{it:th:thetheorem-a}
The  regime is universally optimal.

\item 
\label{it:th:thetheorem-b}
For every state $\statex$ that is not maximal, 
\begin{equation}
\label{eq:Hassin-condition}   
\text{$\posa(\statex)$ assigns probability $0$ to $\njobs(\statex)+1$}.
\end{equation}
\end{enumerate}
\end{theorem}    

\citet{Has:E1985} proved that in an \MMS[1] model, if, for every state $\statex$,
condition \eqref{eq:Hassin-condition} holds, 
then the regime  is universally optimal. 
For example, in \acl{LCFSPR}, \eqref{eq:Hassin-condition} holds for every non-idle state $\statex$.
On the other hand, there exist universally optimal regimes, such as  priority-slots, that do not satisfy \eqref{eq:Hassin-condition} for every state. 
Indeed, under priority-slots, if the system is at a state with $\njobs$ customers occupying  slots $[\njobs]$, then  an arriving customer would be placed at slot $\njobs+1$, \ie at the last position in the system. 
By \cref{ex:priority-slots-2}, these states are exactly the maximal states. 

In a universally optimal regime, at any given time there is a polite customer in the system, \ie a customer who will never be served before the other existing customers in the system.
Except at maximal states, the polite customer is never served before newly arriving customers.
In these states a new polite customer may be chosen. 
As mentioned in the Introduction, the designer must choose the optimal threshold without knowing the parameters of the model. 
This decision is implemented  gradually, for $\statex\in\states_{\njobs}$, with $\njobs=1,2,\dots$, by checking whether a polite customer is willing to stay in the system at the last position $\njobs$, say.
A polite customer who accepts  position $\njobs$ will remain polite until the decision about moving the threshold to $\njobs+1$ has to be made.
This decision is  mediated by a polite customer who---depending on the regime---at threshold $\njobs+1$ may be the same as the one at threshold $\njobs$, or may be different.

We use \cref{th:thetheorem} to formalize the intuition that preemption is necessary for universal optimality.
We first formally define preemption and reassignment in the case of several servers. 
The term reassignment indicates that a customer who is served by some server is moved to another server. 
There are two types of reassignments, which play different roles in our argument. 
The first type, an upgrade, occurs when a customer is reassigned to a better server. 
Such reassignments are necessary in any socially optimal allocation of service, regardless of equilibrium and incentive concerns: Consider the situation where Chantal and Kevin are the only customers in the system, Chantal is being served by the best server and Kevin by the second best server.
If Chantal's service is completed before Kevin's, then efficiency requires that Kevin be upgraded to the better server. 
Recall that our definition of regime already assumes that better servers are always  active if there are enough customers in the system, which implies that an upgrade must occur. 

The second type, a downgrade, indicates a situation where a customer is reassigned to a worse server  upon arrival of a new customer. 
This type of reassignment is more interesting from the perspective of this paper because it is not needed to obtain an optimal service allocation, but rather, it is necessary for an optimal allocation to be achieved in equilibrium. 

Formally, for a given regime, we say that a \emph{downgrade} occurs at server $\labserv$ in state $\statex$ if $\njobs(\statex) \ge \labserv$ and $\posa(\statex)\le \labserv$. 
This means that all customers at server $\posa(\statex)$ or worse (including the customer at server $\labserv$) are being downgraded. 
We say that \emph{preemption occurs in state $\statex$} if a downgrade occurs at server $\nserv$ in state $\statex$.

\begin{corollary}
If a regime is universally optimal, then for any $\labserv\in\braces*{1,\dots,\nserv}$ there exists some state $\statex$ such that a downgrade occurs at server $\labserv$ in state $\statex$.
\end{corollary}

\begin{proof}
Fix a queueing regime and consider the state reached from the idle state after  $\labserv+1$ consecutive arrivals and then a service by some server $\labcustalt \le \labserv$ (for instance, $\labcustalt=1$). 
The new state $\statex$ has $\labserv$ customers and  
is not maximal because it is reached from a state with $\labserv+1$
customers via a service completion that does not serve the last customer.
By \cref{th:thetheorem}, if the regime is universally optimal, then $\posa(\statex)\le \labserv$, \ie a downgrade in state~$\statex$ occurs.
\end{proof}

\begin{remark}
\label{re:preemption}
A mechanism without preemption can be arbitrarily worse than the corresponding mechanism with preemption.
Consider the no-preemption variant of the priority-slots regime in an \MMS[1] model.
In this variant, an arriving customer who joins the queue is assigned to the free slot with the smallest index, but the server continues serving the customer in service and moves to the customer occupying the slot with the smallest index only at the end of the service.
 
Pick $\varepsilon>0$ arbitrarily small and, without  loss of generality, take $\servr=1$.  
Assume that $\reward$ and $\cost$ satisfy  $\reward=2\cost+\varepsilon$. 
In equilibrium, arriving customers will take slot $1$ if it is vacant, and slot $2$ if slot $1$ is occupied but $2$ is vacant. 
If slots $1$ and $2$ are both  occupied, arriving customers  will balk.
When the arrival rate $\arrivr\to\infty$, an arriving customer will typically find at least one of the first two slots occupied, so a joining customer will  remain in the system for two time units in expectation before completing service, and incur a cost of $2\cost$ for a reward $2\cost+\varepsilon$.
Thus, this customer's expected contribution to the social welfare will be $\varepsilon$. 
On the other hand, in this situation,  the social optimum threshold is $\nopt=1$, with expected contribution to the social welfare $\reward-\cost=\cost+\varepsilon$.
This implies that the price of anarchy, \ie the ratio of the optimum social welfare and the equilibrium social welfare, has no finite upper bound. 
\end{remark}  

Because the service time is exponential, the social designer does not need to use preemption to achieve the social optimum. 
Indeed, from a social perspective,  all customers are identical regardless of the time they spent in the system or being served. 
Therefore, the role of preemption is purely strategic, in the sense that it affects the customers' equilibrium behavior.

 
We now use \cref{th:thetheorem} to give two examples of new universally optimal regimes. 
In both examples, the assumption that customers observe the state is essential.

\begin{example}
\label{ex:either-1-or-n+1}   
Consider a regime in which the state is the pair $(\njobs,\maxnjobs)$, where $\maxnjobs$ is the maximum number of customers that the system has seen at any given time since the last idle period.
In this regime a newly arriving customer is immediately sent to position $1$ (\ie is immediately served), whenever $\njobs<\maxnjobs$; otherwise the customer is assigned to the last position. 
That is,
\begin{equation}
\begin{aligned}
\posa(\njobs,\maxnjobs)
&=
\begin{cases}
1 & \text{if }\njobs<\maxnjobs,\\
\njobs+1 & \text{if }\njobs=\maxnjobs,
\end{cases}\\
\transa(\njobs,\maxnjobs)
&=\begin{cases}
(\njobs+1,\maxnjobs) & \text{if }\njobs<\maxnjobs,\\
(\njobs+1,\maxnjobs+1) & \text{if }\njobs=\maxnjobs,
\end{cases}\\
\transs_{\labserv}(\njobs,\maxnjobs)
&= 
\begin{cases}
(\njobs-1,\maxnjobs) & \text{if }\njobs>1,\\ 
 
(0,0) & \text{if }\njobs=1,
\end{cases}
\\
\transr_{\labcust}(\njobs,\maxnjobs)
&= 
\begin{cases}
(\njobs-1,\maxnjobs) & \text{if }\njobs>1,\\ 
(0,0) & \text{if }\njobs=1.
\end{cases}
\end{aligned}
\end{equation}
\end{example}

In this regime, the pair $(\njobs,\maxnjobs)$ is maximal if and only if $\njobs=\maxnjobs$: 
If $\njobs<\maxnjobs$, the state is reachable from $(\maxnjobs,\maxnjobs)$ by $\maxnjobs-\njobs$ services of non-last customers; if $\njobs=\maxnjobs$, it cannot be reached from a larger state without serving the last customer, because the second coordinate records the maximum since the last idle period. 
Hence, the regime in \cref{ex:either-1-or-n+1} is universally optimal because it satisfies the conditions of \cref{th:thetheorem}.

Several deterministic or stochastic  variations of this regime can be constructed. 
For instance, consider a regime with the same state as the one in \cref{ex:either-1-or-n+1}, where a newly arriving customer is placed at random in any position smaller than $\njobs+1$, when $\njobs<\maxnjobs$, and at random in any position (including $\njobs+1$), when $\njobs=\maxnjobs$. 
This regime is also universally optimal.

Under the regime described in the next example, a customer who enters the system receives a score equal to the number of customers who were in the system before this customer entered. 
Customers in the system receive priority based on their score, with lower score customers being served first. 

\begin{example}
\label{ex:score-regime}
A state  $\statex\in\states_{\njobs}$ is given by an $\njobs$-tuple $(\score_{1},\dots,\score_{\njobs})$ with $0\le \score_{1}\le \dots\le \score_{\njobs}$ and $\score_{\labcust}\ge \labcust-1$ for every $\labcust\in [\njobs]$. 
The interpretation is that $\score_{\labcust}$ is the score of the customer who is currently in position $\labcust$. 
The condition $\score_{\labcust}\ge \labcust-1$ reflects the fact that in any group of  $\labcust$ customers at least one of them has a score at least $\labcust-1$ (because the last among them to arrive saw at least all the others).
The transitions  are given by
\begin{equation}
\begin{aligned}
\posa(\score_{1},\dots,\score_{\njobs})
&=\labcuststar+1,\\
\transa(\score_{1},\dots,\score_{\njobs})
&=(\score_{1},\dots,\score_{\labcuststar}, \njobs, \score_{\labcuststar+1},\dots,\score_{\njobs}),\\
\transs_{\labserv}(\score_{1},\dots,\score_{\njobs})
&= (\score_{1}, \dots, \score_{\labserv-1}, \score_{\labserv+1},\dots, \score_{\njobs}),\\
\transr_{\labcust}(\score_{1},\dots,\score_{\njobs})
&= (\score_{1},\dots,\score_{\labcust-1},\score_{\labcust+1},\dots,\score_{\njobs}),
\end{aligned}
\end{equation}
where $\labcuststar=\max\{\labcust\in[\njobs]:\score_{\labcust} < \njobs\}$ and the maximum is defined as $0$ if the set is empty. 
Therefore, $\labcuststar$ is the rank of the last customer who has priority over the incoming customer. 
\end{example}

According to the regime in \cref{ex:score-regime}, a customer who arrives in state $(\score_{1},\dots,\score_{\njobs})$ is placed in the last position  only if $\score_{\njobs}=\njobs-1$. 
Such a state must be maximal.
Indeed, assume by contradiction that we reached this state $\statex$ from a state with $\njobsalt>\njobs$ customers without reneging. 
When the system has $\njobsalt$ customers the score of the last customer is at least $\njobsalt-1$, and arrivals and service completions of customers who are not  last can only increase the score of the last customer. 
Therefore the score of the last customer in state $\statex$ must be at least $\njobsalt-1>\njobs-1$, a contradiction.
Hence, by \cref{th:thetheorem}, the regime is universally optimal.

\begin{remark}
\label{re:implementation}  
Although any universally optimal regime allows a planner to achieve social optimality in equilibrium without knowing the parameters of the model, customers must know them to play the equilibrium.
This situation is typical in the implementation literature, where the planner knows less than the agents.
We refer the reader to \citet{Mas:CUP1985,MasSjo:HSCW2002} for surveys on implementation.
\end{remark}

%
%

\section{Proof of the Main Result}
\label{se:proof}
%
%

\subsection{The Social Optimum}
\label{suse:preliminaries}

We first characterize the socially optimal threshold. 
For an \MMS[1] queueing system, this characterization is already given by 
\citet{Nao:E1969}; see also \citet[Chapter~2]{HasHav:Kluwer2003}. 
For an \MMS system, we use a slightly different argument, which will be useful in the proof of the main result.

We let $\servr\ \coloneqq\sum_{\labserv=1}^{\nserv}\servr_{\labserv}$, and adopt the convention that $\servr_{\labserv}=0$ for every $\labserv>\nserv$.  
In the sequel we will assume that server $\labserv$ provides service according to a Poisson process with rate $\servr_{\labserv}$, with the understanding that,  when  the number of customers in the system is smaller than $\labserv$, a service of server $\labserv$ is fictitious and does not change the state. 
When the system contains $\labserv$ customers, we call a service completion by server $\labserv$ a \emph{critical service}. 

Consider two designers, Alice and Bob. Alice uses threshold $\njobs$, whereas Bob uses threshold $\njobs-1$. 
The two systems are driven by identically coupled arrival and service processes. 
Start from a time in which Alice has $\njobs$ customers and Bob has $\njobs-1$ customers. 
Let the coupled systems run until Alice experiences a critical service. 
At that moment, Bob experiences a fictitious service. 
Then continue the coupled processes until an arrival brings Alice's system to $\njobs$ customers. 

Let  $\stime_{\njobs}^{\crit}$ be the first time at which Alice has a critical service. 

Until that time, Alice has one more customer than Bob; afterward, the two systems contain the same number of customers. 
Therefore, over this comparison cycle, Alice obtains exactly one additional service reward and incurs one additional unit of flow cost for time $\stime_{\njobs}^{\crit}$. 
Define 
\begin{equation}
\label{eq:T-n-def}  
\ruint_{\njobs} \coloneqq \Expect\bracks*{\stime_{\njobs}^{\crit}}.
\end{equation}
The expected difference between Alice's and Bob's payoff during this comparison cycle is given by
\begin{equation}
\label{eq:difference}    
\differ_{\njobs}\coloneqq \reward-\cost \ruint_{\njobs}.
\end{equation}

We now show that $\differ_{\njobs}$ is weakly decreasing in $\njobs$, and strictly decreasing if $\servr_1>\cdots>\servr_{\nserv}$.
Equivalently, we show that $\ruint_{\njobs}$ is weakly increasing in $\njobs$.

To compare $\ruint_{\njobs}$ and $\ruint_{\njobs-1}$, consider a different coupling. 
Take two systems, Alice's system, starting with $\njobs$ customers, and Bob's system, starting with $\njobs-1$ customers. 
Arrivals are coupled identically. 
Suppose that, before either system experiences a critical service event, Alice has $\njobsA$ customers and Bob has $\njobsB$ customers, with $\njobsA>\njobsB>0$. 
Service completions by servers $\labserv<\njobsB$ are coupled identically. 
Since $\servr_{\njobsB}\ge \servr_{\njobsA}$, we couple the critical service clocks so that, at rate $\servr_{\njobsA}$, both systems experience a critical service event, and, at rate $\servr_{\njobsB}-\servr_{\njobsA}$, only Bob's system experiences a critical service event. 
All remaining noncritical service completions are generated independently.

Under this coupling, Bob experiences a critical service event no later than Alice. 
Therefore $\ruint_{\njobs}\ge \ruint_{\njobs-1}$. 
If the service rates are strictly ordered, $\servr_1>\cdots>\servr_{\nserv}$, then the inequality is strict for every $\njobs>1$, as desired.
It follows that $\nopt$ is such that  $\differ_{\njobs}\ge 0$, for every $\njobs\le\nopt$, and $\differ_{\njobs}\le 0$, for every $\njobs>\nopt$. 

We now express the optimality  conditions using a different short cycle comparison. 
Let $\qlthresh^{\njobs}(\ctime)$ be the number of customers in Alice’s system at time $\ctime$, under threshold $\njobs$, starting from $\qlthresh^{\njobs}(0)=\njobs$.
We define $\stime_{\njobs}^{\return}$ as the first time after the next event, either arrival or service, at which  $\qlthresh^{\njobs}(\ctime)=\njobs$.

Then set
\begin{equation}
\label{eq:tau-tilde}
\stimetilde_{\njobs} \coloneqq \min\parens*{\stime_{\njobs}^{\crit},\stime_{\njobs}^{\return}},
\qquad
\ruinptilde_{\njobs} \coloneqq \Prob\parens*{\stime_{\njobs}^{\crit} < \stime_{\njobs}^{\return}},
\qquad
\ruinttilde_{\njobs} \coloneqq \Expect\bracks*{\stimetilde_{\njobs}},
\end{equation}
and let
\begin{equation}
\label{eq:difference-tilde}    \differtilde_{\njobs}\coloneqq \reward\ruinptilde_{\njobs}-\cost\ruinttilde_{\njobs}.
\end{equation}

The number of short cycles until the first one with a critical service has a geometric distribution with success probability $\ruinptilde_{\njobs}$.
The total expected duration until the first critical service is
\begin{equation}
\label{eq:E-tau-c}  
\Expect\bracks*{\stime_{\njobs}^{\crit}}
=\frac{\Expect\bracks*{\stimetilde_{\njobs}}}{\ruinptilde_{\njobs}}
=\frac{\ruinttilde_{\njobs}}{\ruinptilde_{\njobs}}.
\end{equation}
Therefore
\begin{equation}
\label{eq:D-n=}
\differ_{\njobs}
=\reward-\cost \ruint_{\njobs}
=\reward-\cost \frac{\ruinttilde_{\njobs}}{\ruinptilde_{\njobs}}
=\frac{\reward\ruinptilde_{\njobs}-\cost\ruinttilde_{\njobs}}{\ruinptilde_{\njobs}}
=\frac{\differtilde_{\njobs}}{\ruinptilde_{\njobs}},
\end{equation}
which implies that the two quantities $\differ_{\njobs}$ and $\differtilde_{\njobs}$ have the same sign.

As a consequence, the socially optimal threshold $\nopt$ is  also characterized by 
\begin{equation}
\label{eq:D-n-star} \differtilde_{\nopt} \ge 0 \quad\text{and}\quad\differtilde_{\nopt+1} \le 0. 
\end{equation}

\begin{remark}
\label{re:pi-T}
When $\nserv=1$, the event that Alice has a critical service event is the event that Alice's system empties. 
In this case, the values $\ruinptilde_{\njobs}$ and $\ruinttilde_{\njobs}$ are, respectively, the ruin probability and the expected time until the game is over in a gambler's ruin problem with a biased coin. 
\citet[section 2.3]{HasHav:Kluwer2003} used the explicit expressions for these quantities to calculate the optimal threshold $\nopt$. 
These expressions are not needed for our argument.
\citet[section~4.4.2]{HavOz:QMSM2018} used a similar probability-of-ruin argument to compute the social cost of deviation and to find the socially optimal threshold. 

For $\nserv>1$,  we do not have explicit expressions for these quantities, but they can be expressed---using standard Markov chain arguments---as solutions of linear equations in $\njobs$ variables (the number of states in the Markov process). 
\end{remark}

\begin{remark}
\label{re:generic}
For a generic set of parameters $\differ_{\njobs}\neq 0$ for every $\njobs$, so $\nopt$ is unique. 
Moreover, if $\servr_1>\cdots>\servr_{\nserv}$ (in particular, this holds when $\nserv=1$), then there can be at most two optimal thresholds since there can be at most one $\njobs$ with $\differ_{\njobs}=0$.
\end{remark}

%
%

\subsection{Auxiliary Lemmata}

\label{suse:auxiliary-lemmata}

The proof of \cref{th:thetheorem} uses two lemmata, which are of independent interest.

The first lemma connects the payoffs from the socially optimal strategy to the payoffs obtained by an individual customer under any regime.

\begin{lemma}
\label{le:remain}
Consider the following strategy for customer  Chantal:
she remains in the system as long as her position is at most $\njobs$, and reneges immediately if her position becomes larger than $\njobs$.
If $\differtilde_{\njobs}\ge 0$, then, under any regime and any strategy of the opponents, this strategy gives Chantal a nonnegative payoff. 
If $\differtilde_{\njobs} > 0$, then this strategy gives Chantal a strictly positive payoff. 
In addition, if $\differtilde_{\njobs}\ge 0$,  $\servr_1>\cdots>\servr_{\nserv}$, and there is a positive probability that an arriving customer is placed behind Chantal before she is served or reneges, then this strategy gives Chantal a strictly positive payoff.
\end{lemma}

To prove \cref{le:remain}, we now consider the payoff to customer Chantal, under a given regime and a fixed strategy profile of her opponents. 
The proof compares Chantal's position in the actual regime with the number of customers in the system designed by Alice, the designer from the comparison cycle.
Chantal faces a stopping problem: at every time at which a decision is made, she decides whether to stay or to renege. 
The process continues until Chantal is served or reneges.

\begin{proof}[Proof of \cref{le:remain}]
Let the time $\ctime$ fall in the short comparison cycle that defines $\ruinptilde_{\njobs}$ and $\ruinttilde_{\njobs}$, starting from $\qlthresh^{\njobs}(0)=\njobs$, and let $\qlthresh^{\njobs}(\ctime)$ be the number of customers in Alice's system at time $\ctime$. 
By definition, $\qlthresh^{\njobs}(\ctime)\le \njobs$ during the short cycle.

We now couple Chantal's system with Alice's auxiliary system as follows. 
The arrival processes are coupled identically. Reneging only occurs in Chantal's system.  
Call  $\posit(\ctime)$ Chantal's position at time $\ctime$. 
We have that $\posit(\ctime)\le \qlthresh^{\njobs}(\ctime)$ at every time $\ctime$ at which Chantal is in the system.
We now turn to coupling  service completions. 
Service completions by servers $\labserv<\posit(\ctime)$ are coupled identically. 
Chantal's service event occurs at rate $\servr_{\posit(\ctime)}$, whereas Alice's critical service event occurs at rate $\servr_{\qlthresh^{\njobs}(\ctime)}$. 
Since $\posit(\ctime)\le \qlthresh^{\njobs}(\ctime)$, we have $\servr_{\posit(\ctime)}\ge\servr_{\qlthresh^{\njobs}(\ctime)}$. 
We couple these service clocks so that, at rate $\servr_{\qlthresh^{\njobs}(\ctime)}$, Chantal is served and Alice has a critical service event, and, at rate $\servr_{\posit(\ctime)}-\servr_{\qlthresh^{\njobs}(\ctime)}$, Chantal is served but Alice does not have a critical service event. 
All remaining noncritical service completions are generated independently.

Under this coupling, Chantal is served no later than Alice has a critical service event. 
Hence, in each short cycle, the probability that Chantal is served is at least $\ruinptilde_{\njobs}$, and the expected time during which she pays the holding cost before either being served or the short cycle ends is at most $\ruinttilde_{\njobs}$. 
Therefore, her expected payoff in each short cycle is at least
\begin{equation}
\label{eq:D-tilde=}    
\reward\ruinptilde_{\njobs}-\cost\ruinttilde_{\njobs}
=
\differtilde_{\njobs}.
\end{equation}

If the service rates are strictly ordered and  there is a positive probability that an arriving customer is placed behind Chantal before she is served or reneges, then there is a positive probability that  $\posit(\ctime)<\qlthresh^{\njobs}(\ctime)$ before Chantal is served or reneges. 
On this event, there is positive probability that the strict inequality is preserved until Chantal reaches an active server. 
At that time, Chantal's service rate is strictly larger than Alice's critical service rate; therefore, with positive probability Chantal is served before Alice has a critical service event. 
Thus, the above inequality  is strict in some short cycle with positive probability.

It follows that Chantal's expected payoff under the strategy in the statement of the lemma is nonnegative if $\differtilde_{\njobs}\ge 0$, and is strictly positive if either $\differtilde_{\njobs}>0$, or the service rates are strictly ordered, $\differtilde_{\njobs}\ge 0$, and with positive probability an arriving customer is placed behind her before she is served or reneges.  
\end{proof}

The second lemma shows that, if it is socially optimal for a customer to stay in the system, then it is also rational for this customer to do so. 
The lemma also implies an extension of Naor's result to arbitrary regimes: In equilibrium, customers will only start leaving the queue when the number of customers exceeds the social optimum. 
A similar result was proved by \citet{LipSti:OR1977} under general conditions for a regime that does not allow reneging \citep[see also]
[]{HasSni:OR2020}.

\begin{lemma}
\label{le:stay-dominant}
Under every regime, for every customer whose position in the system is not larger than $\nopt$, immediate reneging is weakly dominated by staying and then following the strategy of \cref{le:remain}. 
If $\differtilde_{\nopt}\parameters> 0$, then it is strictly dominated. 
\end{lemma}

\begin{proof}
A customer who reneges gets a continuation payoff equal to $0$. 
If $\njobs\le \nopt$, then $\differtilde_{\njobs}\ge 0$ by the characterization of $\nopt$. 
Therefore, the assertion follows from \cref{le:remain}. 
If $\differtilde_{\nopt}>0$, then $\differtilde_{\njobs}>0$ for every $\njobs\le\nopt$, and the strict assertion follows from \cref{le:remain}.
\end{proof}

%
%

\subsection{The Proof}
\label{suse:proof}

We now have all the tools we need to prove \cref{th:thetheorem}.

\begin{proof}[Proof of \cref{th:thetheorem}]
\noindent
\ref{it:th:thetheorem-b} $\implies$  \ref{it:th:thetheorem-a}. 
Fix an environment $\parameters$ and let $\nopt=\nopt\parameters$.
Let $\stratopt$ be the Markovian strategy profile given by $\stratopt(\statex)=[\njobs(\statex)]\setminus[\nopt]$. 
According to $\stratopt$, all customers whose position is at most $\nopt$ remain and all other customers renege. 
By definition, $\stratopt$ induces the socially optimal behavior. 
We claim that it is a Markov perfect equilibrium, \ie that starting from any state $\statex$, each customer has an incentive to follow $\stratopt$ if the other customers do. 
By \cref{le:stay-dominant}, for a customer whose position is at most $\nopt$ it is a dominant strategy to remain in the system. 
It remains to show that customers at positions above $\nopt$ should renege.  

Fix a customer, Kevin, whose position is higher than $\nopt$. 
Consider a continuation strategy for Kevin that does not renege immediately.
Assume, moreover, that Kevin plans to remain in the system whenever his position is at most $\nopt$, as staying is weakly dominant at such positions.  

Consider the successive times at which either Kevin is served or the system has $\nopt+1$ customers. 
At the first such time, if Kevin has not been served, all opponents whose position is higher than $\nopt$ renege, and therefore Kevin is in position $\nopt+1$.

Starting from a time at which the system has $\nopt+1$ customers and Kevin is in position $\nopt+1$, consider the evolution until the next time at which either Kevin is served or the system has $\nopt+1$ customers. 
If the first event is an arrival, then some opponent, possibly the arriving customer, is in a position higher than $\nopt$ and reneges immediately, while all other opponents remain. 
Thus, if Kevin remains in the system, the system again has $\nopt+1$ customers and Kevin is the last customer.

If the first event is a service completion, then until the system next has $\nopt+1$ customers or Kevin is served, the system has at most $\nopt$ customers, and, by the strategy profile of the opponents, nobody reneges. 
As long as Kevin has not been served, every state reached before the system next has $\nopt+1$ customers is reached from the previous state with $\nopt+1$ customers by a sequence of arrivals and service completions that never serves the last customer.
Therefore, every such state is not maximal. 
Hence, by \ref{it:th:thetheorem-b}, every arriving customer is placed ahead of Kevin. 
Thus, if Kevin remains in the system until the system next has $\nopt+1$ customers, then Kevin is again the last customer.

Therefore, in each such cycle, Kevin can be served only if the last customer in the system is served, and every arriving customer is placed ahead of him. 
Hence, his probability of being served during the cycle is at most $\ruinptilde_{\nopt+1}$, and the expected time during which he pays the holding cost before either being served or the cycle ends is at least $\ruinttilde_{\nopt+1}$. 
This implies that  Kevin's expected payoff during each cycle is at most
\begin{equation}
\label{eq:Kevin's-E-P}   \reward\ruinptilde_{\nopt+1}-\cost\ruinttilde_{\nopt+1}
=
\differtilde_{\nopt+1}
\le 0.
\end{equation}
It follows that Kevin's payoff under any continuation strategy that does not renege immediately is nonpositive, so his best response is to renege immediately.\footnote{Under a generic set of parameters $\differtilde_{\nopt+1}< 0$, so reneging immediately is Kevin's unique best response, proving that it is a strict Nash equilibrium.}

\noindent
\ref{it:th:thetheorem-a} $\implies$  \ref{it:th:thetheorem-b}. 
Assume that a regime $\regime$ does not satisfy \ref{it:th:thetheorem-b}. 
Then there exists a state $\statex$ that is not maximal and such that 
$\posa(\statex)$ assigns a positive probability to $\njobs(\statex)+1$.
Because $\statex$ is not maximal, there exist a sequence of states $\statex_{0},\statex_{1},\dots,\statex_{\round}=\statex$ such that, for $1\le \labcustalt \le \round$, either
$\statex_{\labcustalt}=\transa(\statex_{\labcustalt-1})$ or 
$\statex_{\labcustalt}=\transs_{\labserv}(\statex_{\labcustalt-1})$ for some $\labserv < \njobs(\statex_{\labcustalt-1})$, and $\njobs(\statex_{0}) > \njobs(\statex)$. 
We can also assume without loss of generality that $\njobs(\statex_{0}) > \njobs(\statex_{\labcustalt})$ for $1\le \labcustalt \le \round$. 
Among all states in the witnessing sequence with maximal population size, choose the last one and relabel it as $\statex_{0}$. 
Then all later states in the sequence have strictly fewer customers than $\statex_{0}$, while the sequence still ends at $\statex$ and still uses only arrivals and services that do not serve the last customer.

We will find an open set of parameters under which the social optimum is $\njobs(\statex_{0})-1$, but, in the game that starts at $\statex_{0}$, reneging immediately is not a best response for any customer.

Choose parameters $\ngparameters$ with strictly ordered service rates such that $\differ_{\njobs(\statex_{0})}\ngparameters=0$. 
Because $\differ_{\njobs}$ is strictly decreasing in $\njobs$ when the service rates are strictly ordered, we have $\differ_{\njobs(\statex_{0})-1}\ngparameters>0$. 
Since $\differ_{\njobs}$ and $\differtilde_{\njobs}$ have the same sign, it follows that $\differtilde_{\njobs(\statex_{0})}\ngparameters=0$ and $\differtilde_{\njobs(\statex_{0})-1}\ngparameters>0$. 
Pick parameters $\parameters$ that are sufficiently close to $\ngparameters$ to satisfy
$\differtilde_{\njobs(\statex_{0})-1}\parameters>0$ and $\differtilde_{\njobs(\statex_{0})}\parameters<0$. 
This condition implies that $\nopt\parameters=\njobs(\statex_{0})-1$. 

We now show that, in a game that starts at $\statex_{0}$, no customer wants to renege immediately. 
By \cref{le:stay-dominant}, it is a strictly dominant strategy for all customers except possibly the last one to remain in the system. 
It remains to show that this is the case also for the last customer. 
Call this customer Chantal.

Consider the strategy under which Chantal remains in the system for as long as her position is not larger than $\njobs(\statex_{0})$. 
We first evaluate this strategy under the parameters $\ngparameters$. 
There is a strictly positive probability that Nature generates the events that transition the system along the sequence $(\statex_{1},\dots,\statex_{\round})$, and then an arrival occurs. 
If no arriving customer was placed behind Chantal before the system reached $\statex$, then Chantal is still last at $\statex$. 
Since $\posa(\statex)$ assigns positive probability to $\njobs(\statex)+1$, the next arriving customer is placed behind Chantal with positive probability. 
Thus, with positive probability, an arriving customer is placed behind Chantal before Chantal is served or reneges. 
By \cref{le:remain}, Chantal's payoff from this strategy is strictly positive under $\ngparameters$.

Chantal's payoff under this fixed strategy is continuous in the parameters. 
Therefore, for $\parameters$ sufficiently close to $\ngparameters$, Chantal's payoff from this strategy is still strictly positive. 
Since reneging immediately gives payoff $0$, Chantal does not want to renege immediately.

Thus, under the parameters $\parameters$, the socially optimal threshold is $\njobs(\statex_{0})-1$, but in the game that starts at $\statex_{0}$ no customer wants to renege immediately. 
Therefore the socially optimal behavior cannot be induced by a Markov perfect equilibrium. 
Hence the regime is not universally optimal.
\end{proof}

\begin{remark}
\label{re:unique-optimal-regime}
We have shown that, if a regime satisfies our condition, then it admits a socially optimal equilibrium. 
According to this equilibrium, when the number of customers reaches $\nopt+1$, the last customer reneges. 
In fact, in the generic case with a unique socially optimal threshold $\nopt$, this is also the unique socially optimal equilibrium. 
Indeed, it is socially optimal for some customer to renege when the number of customers reaches $\nopt+1$, and by \cref{le:stay-dominant} the first $\nopt$ customers will remain in an equilibrium. 
Therefore the customer who reneges must be the last in the system.
\end{remark}

\section{Conclusions}
\label{se:conclusions}

Since the early work of \citet{Nao:E1969} a large literature has studied the behavior of queueing systems where customers act strategically. 
Naor showed that, in an \MMS[1] queue under an \ac{FCFS} regime, the equilibrium outcome produced by customers’ selfish rational behavior is socially inefficient. 
He suggested tolls as a way to overcome this inefficiency.
Unfortunately, tolls require an exact calibration that depends on the stochastic parameters of the queue and on the customers’ reward and waiting cost.
For the same \MMS[1] model, 
\citet{Has:E1985} emphasized the role of the queueing regime in achieving efficiency, and proposed regimes that are efficient without requiring a fine calibration.
None of the regimes that he proposed puts a newly arriving customer in the last position.
This condition by itself is sufficient to guarantee efficiency.
Other, more recently proposed, universally optimal regimes do not satisfy this sufficient condition.
To our knowledge, this is the first paper to characterize the whole class of universally optimal regimes in a more general \MMS model.  
Our characterization is based on a suitable twist of Hassin’s sufficient condition. 

An implication of our characterization is that universal optimality requires the use of preemption at least in some situations. 
If preemption is forbidden, then efficiency may be lost. 
In most real-world settings preemption cannot be implemented and is not considered socially acceptable. 
\cref{re:preemption} shows that, when  preemption is forbidden, the loss of efficiency can be severe in some instances, and the price of anarchy is unbounded, even in an \MMS[1] model.

%
%
\section{List of symbols}
\label{se:symbols}


\begin{longtable}{p{.11\textwidth} p{.85\textwidth}}

$\cost$ & cost rate\\
$\differ_{\njobs}$ & $ \reward-\cost \ruint_{\njobs}$, defined in \cref{eq:difference}\\
$\differtilde_{\njobs}$ & $\reward\ruinptilde_{\njobs}-\cost\ruinttilde_{\njobs}$, defined in~\cref{eq:difference-tilde}\\ 

$\labels$ &  $\braces*{\labcust_{1},\labcust_{2}, \dots,\labcust_{\round}} \subset [\njobs]$\\

$\njobs$ & number of customers in the system\\
$\nopt$ & optimum threshold\\
$[\njobs]$ & $\braces*{1,\dots,\njobs}$\\
$\posit(\ctime)$ & Chantal's position at time $\ctime$\\
$\reward$ & reward\\
$\regime$ & regime\\
$\ruint_{\njobs}$ & 
$\Expect\bracks*{\stime_{\njobs}^{\crit}}$\\
$\ruinttilde_{\njobs}$ & $ \Expect\bracks*{\stimetilde_{\njobs}}$, defined in \cref{eq:tau-tilde}\\
$\statex$ & state\\
$\states$ & state space\\
$\states_{\njobs}$ &  set of possible states when there are $\njobs$ customers  in the system\\
$\states_{0}$ & idle system\\
$\qlthresh^{\njobs}(\ctime)$ &  number of customers in Alice's system at time $\ctime$, under threshold $\njobs$, starting from $\qlthresh^{\njobs}(0)=\njobs$\\
$\slot$ & slot\\

$\transa(\statex)$ & new state after a new customer arrives\\
$\score_{\labcust}$ & score of player in position $\labcust$, defined in \cref{ex:score-regime}\\
$\dirac_{\labcust}$ & Dirac measure on $\labcust$\\ 
$\simplex(\setA)$ & set of probability measures on $\setA$\\
$\ruinptilde_{\njobs}$ & $\Prob\parens*{\stime_{\njobs}^{\crit} < \stime_{\njobs}^{\return}}$, defined in \cref{eq:tau-tilde}\\
$\labserv$ & generic server\\
$\arrivr$ & arrival rate\\
$\servr$ & $\sum_{\labserv=1}^{\nserv}\servr_{\labserv}$\\
$\servr_{\labserv}$ & service rate of server $\labserv$\\
$\transs_{\labserv}(\statex)$ & new state after service by server $\labserv$\\
$\posa(\statex)$ & position in the system of the arriving customer\\
$\transr_{\labcust}(\statex)$ & new state after the customer whose current position is $\labcust$ leaves the system\\
$\transr_{\labels}$ & $\transr_{\labcust_{1}} \circ \transr_{\labcust_{2}} \circ \dots \circ \transr_{\labcust_{\round}}$\\
$\strat^{\ast}$ & Markovian strategy profile\\
$\stime_{\njobs}^{\crit}$ & first time Alice has a critical service\\
$\stime_{\njobs}^{\return}$ & first time there is an event (either arrival or service) and $\qlthresh^{\njobs}(\ctime)=\njobs$\\
$
\stimetilde_{\njobs}$ & $ \min\parens*{\stime_{\njobs}^{\crit},\stime_{\njobs}^{\return}}$, defined in \cref{eq:tau-tilde}\\
\end{longtable}

\subsection*{Acknowledgments}
The authors thank Rafi Hassin, Moshe Haviv, and Benny Oz for their useful comments.
Marco Scarsini acknowledges the support of the GNAMPA project CUP\_E53C22001930001 and the MIUR PRIN grant 2022EKNE5K.

\bibliographystyle{apalike}
\bibliography{bib-strategic-queues}

\end{document}